\documentclass[journal]{IEEEtran}

\usepackage{cite}
\usepackage{amsmath, amsthm, amssymb}
\usepackage{algorithm}
\usepackage{algpseudocode}

\usepackage{graphicx}
\usepackage{textcomp}
\usepackage{xcolor}
\definecolor{cream}{RGB}{253, 246, 227}
\definecolor{sgreen}{RGB}{0,100,0}
\definecolor{lblue}{RGB}{30,144,255}
\definecolor{Brown}{RGB}{165,42,42}
\usepackage{comment}
\usepackage{bbm}
\usepackage{lipsum}
\usepackage{fullpage}

\usepackage{subfigure}
\usepackage{amssymb}
\usepackage{booktabs}
\usepackage{bm}
\usepackage[utf8]{inputenc}
\usepackage{amsmath}
\usepackage{amsthm}
\usepackage{dsfont}
\usepackage{cite}
\usepackage{changepage}
\usepackage{lipsum}
\usepackage{caption}
\usepackage{graphicx}
\usepackage{epstopdf}

\long\def\comment#1{}

\newcommand{\be}{\begin{equation}}
\newcommand{\ee}{\end{equation}}

\newtheorem{theorem}{Theorem}

\newtheorem{lemma}[theorem]{Lemma}

\newfont{\bbb}{msbm10 scaled 700}

\newfont{\bb}{msbm10 scaled 1100}

\newcommand{\ev}{{\bf e}}

\newcommand{\sv}{{\bf s}}

\newcommand{\xv}{{\bf x}}

\newcommand{\Dm}{{\bf D}}

\newcommand{\Hm}{{\bf H}}

\newcommand{\Sm}{{\bf S}}

\newcommand{\Ac}{{\cal A}}

\newcommand{\Dc}{{\cal D}}

\newcommand{\Hc}{{\cal H}}

\newcommand{\Kc}{{\cal K}}

\newcommand{\Nc}{{\cal N}}

\newcommand{\Pc}{{\cal P}}

\newcommand{\Sc}{{\cal S}}

\newcommand{\RNum}[1]{\uppercase\expandafter{\romannumeral #1\relax}}

\newcommand{\muv}{\hbox{\boldmath$\mu$}}

\newcommand{\Deltam}{\hbox{\boldmath$\Delta$}}
\newcommand{\Sigmam}{\hbox{\boldmath$\Sigma$}}

\newcommand{\eqdef}{\stackrel{\Delta}{=}}

\usepackage{balance}
\DeclareMathOperator*{\argmin}{arg\,min}
\begin{document}
\title{Stealth Data Injection Attacks with Sparsity Constraints}

\author{Xiuzhen~Ye,
        I\~naki Esnaola, 
       Samir M. Perlaza,
       and Robert~F. Harrison
       
\thanks{
This research was supported in part by the European Commission through the H2020-MSCA-RISE-2019 program under grant 872172 and in part by the China Scholarship Council.

X. Ye, I. Esnaola, and R.~F. Harrison are with the Department of Automatic Control and Systems
Engineering, University of Sheffield, Sheffield S1 3JD, UK.
I. Esnaola is also with the Department of Electrical Engineering, Princeton University, Princeton NJ
08544, USA. (email: xye15@sheffield.ac.uk, esnaola@sheffield.ac.uk, r.f.harrison@sheffield.ac.uk).

S.~M. Perlaza is with the Institut National de Recherche en Informatique
et Automatique (INRIA), Lyon, France, and also with the Department of
Electrical Engineering, Princeton University, Princeton, NJ 08544 USA (email:
samir.perlaza@inria.fr).}
}

%

\maketitle

\begin{abstract}
Sparse stealth attack constructions that minimize the mutual information between the state variables and the observations are proposed. The attack construction is formulated as the design of a multivariate Gaussian distribution that aims to minimize the mutual information while limiting the Kullback-Leibler divergence between the distribution of the observations under attack and the distribution of the observations without attack. The sparsity constraint is incorporated as a support constraint of the attack distribution. Two heuristic greedy algorithms for the attack construction are proposed. The first algorithm assumes that the attack vector consists of independent entries, and therefore, requires no communication between different attacked locations. {The second algorithm considers correlation between the attack vector entries which results in better attack performance at the expense of coordination between different locations}. We numerically evaluate the performance of the proposed attack constructions on IEEE test systems and show that it is feasible to construct stealth attacks that generate significant disruption with a low number of compromised sensors.

\end{abstract}
\section{Introduction}
Monitoring and controlling processes that are supported by supervisory control and data acquisition (SCADA) systems facilitate an economic and reliable operation of the power system\cite{Scada_book}. The integration between the physical layer of the power system and the cyber layer enables efficient, scalable, and secure operation of the system\cite{GJ_PSanalysis_1994}. 
While advanced communication systems that acquire and transmit observations to a state estimator provide reliable and low-latency state information\cite{AA_PSstateestimation_04}, this cyber layer also exposes the system to malicious attacks.  One of the main cybersecurity threats faced by modern power systems {are} data injection attacks (DIAs), which were first introduced in \cite{LY_TISSEC_11}. {DIAs alter the state estimate of the system obtained from different estimation methods by compromising the system observations without triggering bad data detection mechanisms set by the system operator~\cite{NG_Stateestimation_21}}. 
A large body of literature studies the case in which attack detection is performed by a residual test~\cite{VO_SGC_22} under the assumption that state estimation is deterministic both in centralized and decentralized scenarios~\cite{TA_SGC_11},~\cite{CKKPT_SPM_12},~\cite{MO_JSAC_13},~\cite{EPP_gsip_14}.
In this setting, attack construction that requires access to a small set of observations yields $l_0$-norm minimization problems, which are in general hard to solve. In~\cite{KP_TSG_11}, it is shown that the operator can secure a small fraction of observations to make undetectable attack constructions significantly harder. 

The unprecedented data acquisition capabilities that are now available to cyberphysical systems promote the efficient operation of the smart grid but also increase the threat posed by DIAs because accurate stochastic models of the system can be generated. This  problem is cast in a Bayesian framework in \cite{OK_TSG_11}. In this Bayesian paradigm, the attack detection can be formulated as the likelihood ratio test~\cite{IE_TSG_16} or alternatively machine learning methods~\cite{OM_TNNLS_16} can be employed to learn the geometry of the data generated by the system.
Data analytics are increasingly important in the operation of modern power systems and they are central to the advanced estimation, control, and management of the smart grid \cite{TPP_cup_21}.
For this reason, it is essential to study attack constructions in fundamental terms to understand the impact over a wide range of data analysis paradigms.

Stealth data injection attacks within Bayesian framework were first introduced in~\cite{SE_SGC_17} and then generalized in~\cite{SE_TSG_19}. In this research, the attack construction uses information theoretic measures, i.e. mutual information and Kullback-Leibler (KL) divergence, to characterize the fundamental limits of the attack{~\cite{GS_TIT_05}}.
In~\cite{OK_TSG_11}~\cite{SE_SGC_17}~\cite{SE_TSG_19}~\cite{YE_SGC_20}, the state variables are assumed to follow a Gaussian distribution. From a practical point of view, the adoption of Gaussian random vectors as the data injection attack vectors is validated by real data~\cite{GE_spawc_16}~\cite{GE_TSG_18}. 
However, both the stealth attacks constructed in~\cite{SE_SGC_17} and~\cite{SE_TSG_19} require that the attacker tampers with all the observations in the system, which is not feasible in most scenarios. Information theoretic attack constructions that incorporate sparsity constraints are first proposed in~\cite{YE_SGC_20}. However, the construction of attack vectors that effectively exploits the correlation between attack variables is still an open problem that requires novel approaches. In this paper, we present novel sparse stealth attack constructions that leverage the coordination between different attacked observations to attain a better attack disruption to stealth tradeoff.

The rest of the paper is organized as follows: In Section \ref{system model}, we introduce a Bayesian framework with linearized dynamics for DIAs. Stealth attacks incorporating sparsity constraints are presented in Section \ref{sparse construction}. Independent sparse stealth attacks and correlated sparse stealth attacks are presented in Section \ref{ind sparse construction} and Section \ref{corr sparse construction}, respectively. In Section \ref{numerical results}, we evaluate the performance of the proposed attack constructions for both independent and correlated scenarios on IEEE test systems. The paper closes with conclusions in Section \ref{conclusion}.

{The main contributions of this paper follow: 
(1) A novel stealth attack construction with sparsity constraints in Bayesian framework is proposed where the sparse attack is constructed as random attacks.
(2) Information measures are firstly used to construct sparse attacks. Precisely, the attack construction jointly minimizes mutual information and KL divergence.
(3) We tackle the challenge of the combinatorial character of identifying the support of the sparse attack vector by incorporating an additional sensor that yields a sequential sensor selection problem.
(4) Both independent attacks and correlated attacks are considered. In the first case, the random attack requires no communication between locations because its entries are independent. On the other hand, there is correlation between entries in the second case which leads to a better attack performance at the expense of communication. The convexity of the resulting optimization problems in both cases are provided and the insight obtained from incorporating an additional sensor has been distilled to propose heuristic greedy algorithms, accordingly.}

\textbf{Notation:} We denote the number of state variables on a given IEEE test system by $n$ and the number of the observations by $m$. The set of positive semidefinite matrices of size $n\times n$ is denoted by $S_{+}^n$. The $n$-dimensional identity matrix is denoted as $\textbf{I}_n$.
%
%
The elementary vector $\ev_i\in\mathds{R}^n$ is a vector of zeros with a one in the $i$-th entry. Random variables are denoted by capital letters and their realizations by the corresponding lower case, e.g. $x$ is a realization of the random variable $X$. Vectors of $n$ random variables are denoted by a superscript, e.g. $X^n=(X_1, \ldots, X_n)^{\sf{T}}$ with corresponding realizations denoted by $\xv$. 
Given an $n$-dimensional vector $\muv \in \mathds{R}^n$ and a matrix $\Sigmam \in S_{+}^n$, we denote by $\mathcal{N} (\muv, \Sigmam)$ the multivariate Gaussian distribution of dimension $n$ with mean $\muv$ and covariance matrix $\Sigmam$. 
The mutual information between random variables $X$ and $Y$ is denoted by $I(X;Y)$ and the Kullback-Leibler (KL) divergence between the distributions $P$ and $Q$  is denoted by $D(P\| Q)$.

\section{System model}\label{system model}

\subsection{Observation Model and Attack Setting}

The operation state of a power system is described by a vector $\xv \in{\mathds{R}^n}$ containing the voltages and phases at all the generation and load buses. The state vector $\xv$ is observed through the acquisition function $F: {\mathds{R}^n} \rightarrow {\mathds{R}^m}$. When a linearized observation model is considered for state estimation, it yields an observation model of the form
\be\label{eq:obs_noattack}
Y^m  = \textbf{H}\xv+Z^m,
\ee
where $\textbf{H} \in {\mathds{R}^{m \times n}}$ is the Jacobian of the function $F$ at a given operating point and is determined by the system entries and the topology of the network. The vector $Y^m$ containing the observations is corrupted by additive white Gaussian noise introduced by the sensors, c.f.,~\cite{GJ_PSanalysis_1994} and~\cite{AA_PSstateestimation_04}. Such noise is modelled by the vector $Z^m$ in \eqref{eq:obs_noattack}, which follows a multivariate Gaussian distribution. That is, 
\be 
\label{EqZ}
Z^m \sim \mathcal{N}(\textbf{0},\sigma^2 \textrm{\textbf{I}}_m),
\ee 
where $\sigma^2$ is the noise variance.

In a Bayesian estimation framework, the state variables are described by a random vector $X^n$ with a given distribution. In this study, the random vector $X^n$ is assumed to follow a multivariate Gaussian distribution with a null mean vector and covariance matrix 
\begin{equation}
	\label{EqSigmaXX}
	\Sigmam_{X\!X} \in S_{+}^n. 
\end{equation}
Hence,  the vector of observations $Y^m$ in~\eqref{eq:obs_noattack} follows a multivariate Gaussian distribution with null mean vector and a covariance matrix  $ \Sigmam_{Y\!Y}$  satisfying that
\begin{equation}\label{2} 
	\Sigmam_{Y\!Y} \triangleq \textbf{H} \Sigmam_{X\!X}\textbf{H}^{\sf{T}}+\sigma^2 \textrm{\textbf{I}}_m.
\end{equation}
The resulting observations are corrupted by a malicious attack vector 
$A^m \sim P_{A^m}$,  where $P_{A^m}$ is the distribution of the random vector $A^m$. In the following, $P_{A^m}$ is assumed to be a multivariate Gaussian distribution that satisfies
\begin{equation}\label{eq:Gauss_attack}
	A^m \sim \mathcal{N} (\textbf{0}, \Sigmam_{A\!A}),
\end{equation}
where $\textbf{0} = \left( 0,0, \ldots, 0 \right)$ and $\Sigmam_{A\!A}\in S_{+}^m$ are the mean vector and the covariance matrix of the random vector $A^m$. 

The choice in~\eqref{eq:Gauss_attack} is justified by the fact that a multivariate Gaussian distribution minimizes the mutual information between the state variables and the compromised observations under the assumption that the covariance matrix {$\Sigmam_{A\!A}$} is fixed~\cite{SI_TIT_13}.
Consequently, the compromised observations denoted by $Y_A^m$ are given by
\begin{equation}
	\label{eq:obs_attack}
	Y_A^m  = \textbf{H}X^n+Z^m + A^m,
\end{equation}
where $Y_A^m$ follows a multivariate Gaussian distribution given by
\begin{equation}\label{6}
	Y^m_A  \sim \mathcal{N} (\textbf{0}, \Sigmam_{Y_A\!Y_A})
\end{equation}
with $\Sigmam_{Y_A\!Y_A} = \textbf{H}\Sigmam_{X\!X}\textbf{H}^{\sf{T}} + \sigma^2 \textrm{\textbf{I}}_m + \Sigmam_{A\!A}$.

\subsection{Attack Detection}
As a part of a security strategy, the operator implements an attack detection procedure prior to performing state estimation. Detection is cast as a hypothesis testing problem given by
\begin{subequations}\label{EqHypTestA}
	\begin{align}\label{eq:hypoth_attack}
		\mathcal{H}_0&:\textrm{There is no attack,}\\ 
	\mathcal{H}_1&:\textrm{Observations are compromised}.
\end{align}
\end{subequations}
At time step $i\in\mathds{N}$, the system operator acquires a vector of observations $\bar{Y}_i^m$ and decides whether the vector of observations $\bar{Y}_i^m$ is produced following a no attack scenario as described in \eqref{eq:obs_noattack} or is the result of the attack as described in \eqref{eq:obs_attack}. In our setting, the hypothesis test can be recast in terms of the probability density functions induced by the state variables, the system noise, and the attack onto the observations $\bar{Y}^m$. Hence, the hypotheses in \eqref{EqHypTestA} become 
\begin{subequations}\label{eq:hypoth_attack}
\begin{align}
	\mathcal{H}_0&: \bar{Y}^m\thicksim P_{Y^m}, \\
	\mathcal{H}_1&: \bar{Y}^m\thicksim P_{Y^m_A}.
\end{align}
\end{subequations}
A test to determine what distribution generates the observation data is a deterministic test $T:\mathds{R}^m\rightarrow\{0,1\}$. Given an observation vector $\mathbf{\bar{y}}$, let $T(\mathbf{\bar{y}} )=0$ denote the case in which the test decides $\Hc_0$ upon the observation of $\mathbf{\bar{y}}$; and $T(\mathbf{\bar{y}} )=1$  the case in which the test decides $\Hc_1$.
The performance of the test is assessed in terms of the Type-I error, denoted by $\alpha\eqdef \mathds{P}\left[T\left(\bar{Y}^m\right)=1\right]$, with $\bar{Y}^m\thicksim P_{Y^m}$; and the Type-II error, denoted by $\beta\eqdef \mathds{P}\left[T\left(\bar{Y}^m\right)=0 \right]$, with $\bar{Y}^m\thicksim P_{Y_A^m}$.
Given the requirement that the Type-I error satisfies $\alpha \leq  \alpha'$, with $\alpha'\in[0,1]$, the likelihood ratio test (LRT) is optimal in the sense that it induces the smallest Type-II error $\beta$~\cite{JN_LRT_33}. In this setting, the LRT is given by
\begin{equation}\label{lrt}
T(\mathbf{\bar{y}}) = \mathds{1}_{\left\lbrace L(\mathbf{\bar{y}}) \geqslant \tau \right\rbrace},
\end{equation}
with $L(\mathbf{\bar{y}})$ is the likelihood ratio, i.e., 
\begin{equation}\label{lr}
L(\mathbf{\bar{y}}) = \frac{f_{Y_A^m}(\mathbf{\bar{y}})}{f_{Y^m}(\mathbf{\bar{y}})},
\end{equation}
where the functions $f_{Y_A^m}$ and $f_{Y^m}$ are respectively  the probability density function (pdf) of $Y_A^m$ in \eqref{eq:obs_attack} and the pdf of $Y^m$ in~\eqref{eq:obs_noattack}; and $\tau\in\mathds{R}_+$ in~(\ref{lrt}) is the decision threshold. 
Note that changing the value of $\tau$ is equivalent to change the tradeoff between Type-I and Type-II errors.

\section{Sparse Stealth Attacks}\label{sparse construction}

\subsection{Information Theoretic Metric}
The aim of the attacker is twofold. First, it aims to inflict a data integrity attack that disrupts all processes that use the observations of the system; and second, to guarantee a stealthy attack.
Hence, instead of assuming a particular state estimation procedure, we adopt the methodology in~\cite{SE_TSG_19} to construct stealth attacks that minimize the amount of information acquired by the observations about the state variables. In doing so, the attacker targets a universal utility metric consisting in a weighted sum of two terms~\cite{TM_ElementsofIT}: $(a)$ the mutual information between the state variables and the observations; and $(b)$ the KL divergence between the probability distribution functions of the observations with and without attack. By minimizing this metric, the attacker guarantees a stealthy attack that impinges upon any procedure using the observations.

{The KL divergence term guarantees a stealthy attack in the sense that its minimization leads to minimizing the absolute difference between the probability of false alarm and the probability of attack detection, i.e. $|\alpha-(1-\beta)|$~\cite{JN_LRT_33}~\cite{HV_introtosignal}.}

Within this framework, stealth attacks are constructed as random vectors whose probability distribution functions are the solution to the following optimization problem:
\begin{equation}\label{eq:stealth_opt}
  \min_{P_{A^m}} I(X^n;Y^m_A)+ \lambda D(P_{Y_A^m}\|P_{Y^m}),
\end{equation}
where the optimization domain is the set of all possible $m$-dimensional  Gaussian probability distributions; and  $\lambda\geq 1$ is a weighting parameter that determines the tradeoff between the attack disruption and probability of attack detection. 

The solution to the optimization in~\eqref{eq:stealth_opt} is a multivariate Gaussian distribution for the attack vector. It is shown in \cite{SE_TSG_19} that the optimal Gaussian attack is given by $\bar{P}_{A^m} \sim \Nc(\mathbf{0},\bar{\Sigmam})$ where
\be\label{eq:stealth_cov}
\bar{\Sigmam}={\lambda^{-1/2}}\Hm\Sigmam_{X\! X}\Hm^{\sf T}.
\ee
Note that \eqref{eq:stealth_cov} yields a stealth attack vector that is not sparse, indeed all the entries of the attack realizations are nonzero with probability one, i.e. $\mathds{P}\left[|\textnormal{supp}({A^m})|=m\right]=1$, where we define the support of the attack vector ${A^m}$ as
\be
\label{EqSuppA}
\textnormal{supp}({A^m})\eqdef\left\{i:\mathds{P}\left[A_i=0\right]=0\right\}.
\ee
\vspace{-7mm}
\subsection{Sparse Stealth Attack Formulation}

The attack implementation requires access to the sensing infrastructure of the industrial control system (ICS) operating the power system. Data injection attacks usually exploit the vulnerabilities existing in the field zone by comprising remote terminal units or local secondary level control systems, or alternatively, by getting access to the SCADA system coordinating the control zone of the ICS. For that reason, attack constructions that are required to intrude the least amount of monitoring and data acquisition infrastructure are particularly interesting. In view of this, we study sparse attacks that require access to a limited number of sensors, i.e. we pose the attack construction problem with sparsity constraints by setting the domain as the set of distributions over the attack vector that put non-zero mass on at most $k\leq m$ {attack vector entries}.

In our formulation, this is reflected by an additional optimization constraint of the form  $\left| \textnormal{supp}({A^m})\right|=k$, for some given $k \leqslant m$. Hence, the attacker chooses the distribution over the set of multivariate Gaussian distributions given by
\be
\label{EqSparcityDomain}
\Pc_k\eqdef\left \{ P_{A^m}\sim \Nc(\mathbf{0},\bar{\Sigmam}):\left| \textnormal{supp}({A^m})\right | =k\right \}.
\ee
The resulting $k$-sparse stealth attack construction is therefore posed as the optimization problem:
\be\label{eq:k_sparse_stealth_opt}
  \min_{P_{A^m}\in\Pc_k} I(X^n;Y^m_A)+ \lambda D(P_{Y_A^m}\|P_{Y^m}).
\ee

The optimization domain including the sparsity constraint in~\eqref{EqSparcityDomain} implies an additional difficulty in the construction of stealth attacks with respect to the construction proposed in \cite{SE_TSG_19}. This additional difficulty lies on the combinatorial problem arising from the selection of at most $k$ out of $m$ dimensions of the vector attack to form the support of $A^{m}$.
To tackle this difficulty, we exploit the structure that the Gaussian attack embeds into the sparse attack problem formulation to propose novel attack construction algorithms with verifiable performance guarantees. 
\vspace{-4mm}
\subsection{Gaussian Sparse Stealth  Attack Construction}

The probability distribution function of a random vector is determined by two parameters, i.e., the mean vector and the covariance matrix. Hence, writing the objective function of the optimization problems in~\eqref{eq:stealth_opt} and~\eqref{eq:k_sparse_stealth_opt} in terms of the mean vector and covariance matrix of the attack random vector $A^m$ leads to observing that it is equal to the following expression, up to a constant additive term,
\be
\label{eq:Gauss_cost}
\begin{aligned}
J(\Sigmam_{A\!A})\eqdef&(1-\lambda) \log |\Sigmam_{Y\!Y} +\bm{\Sigma}_{A\!A}| \\
&- \log | \sigma^2\textbf{I}_m + \bm{\Sigma}_{A\!A}|  + \lambda \textrm{tr}(\Sigmam^{-1}_{Y\!Y}\bm{\Sigma}_{A\!A}),
\end{aligned}
\ee
where $\lambda \geq 1$ is introduced in \eqref{eq:stealth_opt}; and the matrix $\Sigmam_{Y\!Y}$ is defined by \eqref{2}.

Hence, the optimization problem in~\eqref{eq:stealth_opt} is equivalent to the following optimization problem:
\be
\label{eq:Gaussian_stealth_constr}
\min_{\bm{\Sigma}_{A\!A} \in S_{+}^m }J(\Sigmam_{A\!A}).
\ee
In order to write the optimization domain of the problem in~\eqref{eq:k_sparse_stealth_opt} in terms of the mean vector and covariance matrix of the attack random vector, it suffices to observe that the sparsity constraint in~\eqref{EqSparcityDomain} translates into a constraint on the number of nonzero entries in the diagonal of the covariance matrix of the attack vector. More specifically,  the optimization domain becomes:
\be
\label{eq:cov_sparse}
\Sc_k\eqdef \left\{\Sm\in S_{+}^m: \| \textnormal{diag}(\Sm) \|_0=k \right\},
\ee
where $\textnormal{diag}(\Sm)$ denotes the vector formed by the diagonal entries of $\Sm$.
Solving~\eqref{eq:Gaussian_stealth_constr} within the optimization domain specified by~\eqref{eq:cov_sparse} re-casts the equivalent $k$-sparse stealth attack construction problem in~\eqref{eq:k_sparse_stealth_opt} as:
\be
\label{eq:Gaussian_k_stealth_constr}
\min_{\bm{\Sigma}_{A\!A} \in \Sc_k } J(\Sigmam_{A\!A}).
\ee

\section{Independent Sparse Stealth  Attacks}\label{ind sparse construction}

We first tackle the case in which the attack vector entries are independent. More specifically, the focus is on product probability measures of the form \begin{equation}
\label{EqIndepPA}
P_{A^m}=\prod_{i=1}^m P_{A_i},
\end{equation} 
where, for all $i \in \lbrace 1,2, \ldots, m\rbrace$, the probability density function of the measure $P_{A_i}$ is Gaussian with zero mean and variance $v_i$. 

The assumption of independence relaxes the correlation requirements between the entries of the attack vector.  
As a result, the set of covariance matrices given by \eqref{eq:cov_sparse}, with $k \leqslant m$, that arises from considering Gaussian attacks is the set

\be
\tilde{\Sc}_k\!\eqdef \bigcup_{\Kc}\!\left\{\Sm\!\in\! S_{+}^m\!\!:\Sm\!=\!\!\sum_{i\in\Kc}v_i \ev_i\ev_i^{\sf T}\textnormal{with}\;v_i\!\in\!\mathds{R}_+\!\right\},
\ee
where the union is over all subsets $\Kc\subseteq\left\{1,2, \ldots, m\right\}$ with $|\Kc|=k\leq m$. Note that it holds that $\tilde{\Sc}_k\subseteq \Sc_k$. 

Under the independence assumption adopted in this section, the optimization problem in \eqref{eq:Gaussian_stealth_constr} boils down to the following problem:
\be
\label{eq:Gaussian_k_stealth_indep}
\min_{{\Sigmam}_{A\!A} \in \tilde{\Sc}_k } J(\Sigmam_{A\!A}),
\ee
which is hard to solve due to the combinatorial character of identifying the support of the sparse random attack vector.
To circumvent this problem, we propose a greedy construction that sequentially updates the set $\textnormal{supp}({A^m})$ in~\eqref{EqSuppA} and determines the corresponding entry in the diagonal of the matrix $\bm{\Sigma}_{A\!A}$ in~\eqref{eq:Gauss_attack}.

\subsection{Greedy Independent Attack Construction}

The proposed construction hinges on the idea that approaching the sensor selection problem in a sequential fashion resembles the single sensor selection problem discussed in \cite{YE_SGC_20}. This enables us to leverage the single sensor selection construction to analytically characterize the cost difference induced by the addition of a new element to the set $\textnormal{supp}({A^m})$ in~\eqref{EqSuppA}. 

More specifically, given the sparsity constraint in \eqref{eq:cov_sparse}, for some $k \leqslant m$, the construction can be divided into $k$ epochs. At each epoch a new element is added to $\textnormal{supp}({A^m})$. 
At epoch $i$, let $\Sigmam_i \in S_{+}^m$ be the covariance matrix of the vector attack under construction. 
Let the set $\mathcal{A}_{i}$ be the set of indices corresponding to the entries of the vector $\textnormal{diag}(\Sigmam_i)$ that are different from zero.  
That is,
\begin{equation}
\label{EqSetAi}
\mathcal{A}_i = \lbrace j \in \lbrace 1,2, \ldots, m \rbrace:  \ev_j^{\sf T}  \Sigmam_i \ev_j > 0 \rbrace.
\end{equation}
For all $i \in \lbrace 1,2, \ldots, k \rbrace$, it is imposed that $\mathcal{A}_{i} \subseteq \lbrace 1,2, \ldots, m \rbrace$ and $|\mathcal{A}_{i}| = i$. This implies that  $\mathcal{A}_{1} \subset \mathcal{A}_{2} \subset \ldots \subset \mathcal{A}_{k} \subset  \lbrace 1,2, \ldots, m \rbrace$.
Hence, 
\begin{equation}
\label{eq:greed_step}
\Sigmam_{i} = \Sigmam_{i-1}+ v \ev_j \ev^{\sf T}_j,
\end{equation}
where $\Sigmam_0$ is a matrix of zeros; the integer $j \in \lbrace 1, 2, \ldots, m \rbrace \setminus \mathcal{A}_{i-1}$ is the index of the new entry at epoch $i$; and $v > 0$ is the value of such entry. 
For ease of presentation we denote the set of indices available to the attacker to choose at epoch $i$, i.e. the entries of the vector $\textnormal{diag}(\Sigmam_{i-1})$ that are zero, as
\be
\label{eq:A_comp}
\mathcal{A}_{i-1}^{\sf c}\eqdef \lbrace 1, 2, \ldots, m \rbrace \setminus \mathcal{A}_{i-1}.
\ee
Our proposition to choose both $j \in \mathcal{A}_{i-1}^{\sf c}$ and $v > 0$ at epoch $i$ as described in \eqref{eq:greed_step} is based on the following optimization problem
		\be
		\label{op:indep_sel}
		\min_{(j, v)\in\mathcal{A}_{i-1}^{\sf c}\times\mathds{R}_+} J(\Sigmam_{i-1}+v \ev_j\ev^{\sf T}_j).
		\ee
The following lemma sheds light on the solution to the problem \eqref{op:indep_sel}.
 
\begin{lemma}
\label{lm:cost_diff}
Let $\Sigmam_1\in S_{+}^m$ and $\Sigmam_2\in S_{+}^m$ be two matrices that satisfy $\Sigmam_2=\Sigmam_1+\bm{\Delta}$, with $\bm{\Delta}\in\mathds{R}^{m\times m}$. Then, the cost function $J$ in \eqref{eq:Gauss_cost} satisfies that 
\be
J(\Sigmam_2)=J(\Sigmam_1)+f(\Sigmam_1, \bm{\Delta}),
\ee
where the function $f: \mathds{R}^{m\times m} \times \mathds{R}^{m\times m} \rightarrow \mathds{R}$ is such that 
\begin{IEEEeqnarray}{rcl}
\nonumber
f(\Sigmam_1, \bm{\Delta}) & = & (1-\lambda) \log \left|\textnormal{\textbf{I}}_m +\left(\Sigmam_{Y\!Y}+\Sigmam_1\right)^{-1}\bm{\Delta}\right| \\
\nonumber
& &  - \log \left| \textnormal{\textbf{I}}_m +\left(\sigma^2\textnormal{\textbf{I}}_m+\Sigmam_1\right)^{-1}\bm{\Delta}\right| \\
\label{Eqf}
& & + \lambda \textnormal{tr}\left(\Sigmam^{-1}_{Y\!Y}\bm{\Delta}\right),
 \end{IEEEeqnarray}
where $\lambda \geq 1$ is introduced in \eqref{eq:stealth_opt}; and the matrix $\Sigmam_{Y\!Y}$ is defined by \eqref{2}.
\end{lemma}
\begin{proof}
The proof consists in showing that the difference between $J(\Sigmam_2)$ and $J(\Sigmam_1)$ yields
\begin{IEEEeqnarray}{rcl}
\nonumber
J(\Sigmam_2)-J(\Sigmam_1) &  
= &(1-\lambda) \log \left|\textnormal{\textbf{I}}_m +\left(\Sigmam_{Y\!Y}+\Sigmam_1\right)^{-1}\bm{\Delta}\right| \\
\nonumber
& &- \log \left| \textnormal{\textbf{I}}_m +\left(\sigma^2\textnormal{\textbf{I}}_m+\Sigmam_1\right)^{-1}\bm{\Delta}\right|\\
& &  + \lambda \textnormal{tr}\left(\Sigmam^{-1}_{Y\!Y}\bm{\Delta}\right),
\end{IEEEeqnarray} 
which completes the proof.
\end{proof}

The relevance of Lemma~\ref{lm:cost_diff} is that it enables the selection of both $j \in\mathcal{A}_{i-1}^{\sf c}$ and $v > 0$ at epoch $i$ based on a simpler optimization problem than that in~\eqref{op:indep_sel}. Indeed, the selection problem results in
\be
		\label{op:indep_sel2}
		\min_{(j,v)\in\mathcal{A}_{i-1}^{\sf c}\times\mathds{R}_+} f(\Sigmam_{i-1}, v \ev_j\ev^{\sf T}_j),
		\ee
where the function $f$ is defined in \eqref{Eqf}.		
Theorem \ref{th:indep_sel} provides the solution to the optimization problem in \eqref{op:indep_sel2}.
	\begin{theorem}
		\label{th:indep_sel}
		Let $k$ satisfy $0 < k \leqslant m$,  and for all $i \in \lbrace 1,2, \ldots, k \rbrace$, denote by $\left(j^{\star}, v^{\star} \right) \in \mathcal{A}_{i-1}^{\sf c}\times\mathds{R}_+$
the solution to the optimization problem in  \eqref{op:indep_sel}. Then, the following holds
\begin{IEEEeqnarray}{rcl}
\label{EqJStar}
j^{\star} & = & \argmin_{j \in \mathcal{A}_{i-1}^{\sf c}} J(\Sigmam_{i-1}+v_{j} \ev_{j}\ev^{\sf T}_{j}) \;\mbox{ and }\\ 
\label{EqThetaJStar}
v^{\star}  & = & v_{j^{\star}}, 
\end{IEEEeqnarray}
where, for all $j \in \mathcal{A}_{i-1}^{\sf c}$ 		
\begin{IEEEeqnarray}{l}
\nonumber
\!\!\!\!\!\!v_{j^*}  =  \left(\dfrac{\beta_{j} -\alpha_{j}+\beta_{j} \alpha_{j}\sigma^2 }{2\beta_{j} \alpha_{j}} \right) \\
\label{eq:sigma_i_ind}
 \cdot \!\!\left(\!\!\sqrt{\!1\! - \!\dfrac{4\beta_{j} \alpha_{j} \left(\beta_{j}\sigma^2-\alpha_{j}\sigma^2 - \dfrac{\alpha_{j}\sigma^2+1}{\lambda} \right)}{\left(\beta_{j} -\alpha_{j}+\beta_{j} \alpha_{j}\sigma^2 \right)^2} }\! -\!1\!\right) \,
\end{IEEEeqnarray}
with 
\begin{IEEEeqnarray}{rcl}
\alpha_j &\eqdef & \textnormal{tr}\left(\left(\Sigmam_{Y\!Y}+\Sigmam_{i-1}\right)^{-1} \ev_{j^*}\ev^{\sf T}_{j^*} \right),\\
\beta_j &\eqdef & \textnormal{tr}\left(\Sigmam_{Y\!Y}^{-1}  \ev_{j^*}\ev^{\sf T}_{j^*}\right),
\end{IEEEeqnarray}
and the real $\sigma > 0$ in \eqref{eq:sigma_i_ind} is introduced in \eqref{EqZ}.
	\end{theorem}
	\hrule
\vspace{3mm}
\begin{proof}
 It follows from Lemma \ref{lm:cost_diff} that the optimization problem in \eqref{op:indep_sel2} is equivalent to 
 
 \be
 \begin{aligned}
 	\min_{(j,v)\in\mathcal{A}_{i-1}^{\sf c}\times\mathds{R}_+} &(1-\lambda) \log \left|\textnormal{\textbf{I}}_m +\left(\Sigmam_{Y\!Y}+\Sigmam_{i-1}\right)^{-1}v\ev_{j} \ev_{j}^{\sf T}\right| \\
 	&  - \log \left| \textnormal{\textbf{I}}_m +\left(\sigma^2\textnormal{\textbf{I}}_m+\Sigmam_{i-1}\right)^{-1}v\ev_j\ev_j^{\sf T}\right| \\
 	& + \lambda \textnormal{tr}\left(\Sigmam^{-1}_{Y\!Y}v\ev_j\ev_j^{\sf T}\right).
 \end{aligned}
\ee
 After some algebraic manipulation it follows that 
  \be
 \begin{aligned}\label{att_ind_mini_pro}
 	\min_{(j,v)\in\mathcal{A}_{i-1}^{\sf c}\times\mathds{R}_+} &\!\!(1\!-\!\lambda) \log (1+\alpha_j v)\! -\! \log(1\!+\!\frac{v}{\sigma^2}) \!+\! \lambda \beta_j v,
 \end{aligned}
 \ee
 which is convex for $\lambda \geq 1$. The only solution of the minimization problem in \eqref{att_ind_mini_pro} is obtained by letting the derivative to zero, which yields 
 \be\label{att_ind_deri}
 \beta_j\alpha_jv^2 + (\beta_j-\alpha_j+\beta_j\alpha_j\sigma^2)v +\beta_j\sigma^2-\alpha_j\sigma^2-\frac{\alpha_j\sigma^2+1}{\lambda}\! =\! 0
 \ee
 Note that \eqref{att_ind_deri} is quadratic with two solutions. The result follows by choosing the solution such that $v \in \mathds{R}_+$. This completes the proof.
\end{proof}
The proposed greedy construction is described in Algorithm~\ref{alg:greedy_independent}.
\begin{algorithm}
	\caption{$k$-sparse independent attack construction}
	\label{alg:greedy_independent}
	\begin{algorithmic}[1]
		\Require 	$\Hm$ in \eqref{eq:obs_noattack};\newline
				$\sigma^2$ in \eqref{EqZ};\newline
				$\Sigmam_{X\!X}$ in \eqref{EqSigmaXX};\newline
				$\lambda$ in \eqref{eq:Gauss_cost}; and \newline
				$k$ in \eqref{EqSetAi}.
		\Ensure $\bm{\Sigma}_{A\!A}$ in \eqref{eq:Gauss_attack}.
		\State Set $\Ac_0=\left\{\emptyset \right\}$
		\State Set $\Sigmam_0=\mathbf{0}$
		\For {$j=1$ to $k$}
		\For {$\ell \in\Ac_{i-1}^{\sf c}$}
		\State Compute $v_{\ell}$ in~\eqref{eq:sigma_i_ind}
		\EndFor
		\State Compute $j^{\star}$ in~\eqref{EqJStar}
		\State Compute $v^{\star}$ in~\eqref{EqThetaJStar}
		\State Set $\Ac_j = \Ac_{j-1}\cup \left\{j^{\star}\right\}$
		\State Set $ {\Sigmam}_j = \sum_{i\in\Ac_j}v_{i}\ev_i\ev_i^{\sf T}$ 
		\EndFor
		\State $\bm{\Sigma}_{A\!A} = \sum_{i\in\Ac_k}v_{i}\ev_i\ev_i^{\sf T}$ 
	\end{algorithmic}
\end{algorithm}

\section{Correlated Sparse Stealth  Attacks}\label{corr sparse construction}
\subsection{Correlation Structure}\label{op_dep_att_con}

In this section, the assumption of independence in \eqref{EqIndepPA} is dropped. This case boils down to the attack construction given in \eqref{eq:Gaussian_k_stealth_constr}, i.e. the optimization is carried over the set of covariance matrices with non-zero off-diagonal entries that account for the correlation between different attack entries. 
In this case the addition of a new index to the set of $k$ attacked observations introduces off-diagonal entries in the difference between covariance matrices described in Lemma \ref{lm:cost_diff}. More precisely, the difference introduced by selecting the index $i$ is given by
$
\Delta_i\in\Dc_i
$
with
\be
\Dc_i=\bigcup_{\sv\in\mathds{R}^m}\left\{\Dm\in\mathds{R}^{m\times m}: \Dm=\sv^{\sf T}\!\otimes \ev_i + \sv \otimes \ev_i^{\sf T}, \right\}.
\ee
Note that the vector $\sv$ determines the second order moments describing the covariance between attacked observations.
As in the independent case, characterizing the difference enables to formulate the optimization problem that yields the minimum cost increase introduced by a new index in the attack support. Let $\Ac_{k-1}$ denote set of indices of attacked observations and $\Sigmam_{i-1}\in\Sc_{i-1}$ the covariance matrix of the attack vector over those $i-1$ observations. Then the sensor selection problem at step $i$ is given by the optimization problem:
\begin{equation}
	\begin{aligned}\label{eq:opt_corr}
		 &\min_{j,\Deltam}   \ J\left(\Sigmam_{i-1}+\Deltam\right)\\
		 &\ \textnormal{s.t.} \ \ \ {j\in\Ac_{i-1}^{\sf c}},\\
		&\qquad \  \Deltam\in\Dc_j, \\
		&\qquad \ \Sigmam_{i-1}+\Deltam\in S_+^m. \\
	\end{aligned}
\end{equation}

In the following we show that when the choice of the next index selected for attacks is fixed, 
the optimization in \eqref{eq:opt_corr} is convex in the matrix difference. 
\begin{theorem}
\label{th:delta_conv_corr}
Let $\Sigmam_{i-1}\in\Sc_{i-1}$ and $j\in\Ac_{i-1}^{\sf c}$, then the optimization problem given by
\begin{equation}
	\begin{aligned}\label{eq:delta_conv_corr} 
		&\min_{\Deltam}   \ J\left(\Sigmam_{i-1}+\Deltam\right)\\
		&\ \textnormal{s.t.} \ \ \  \Deltam\in\Dc_j, \\
		&\qquad \ \ \Sigmam_{i-1}+\Deltam\in S_+^m, \\
	\end{aligned}
\end{equation}
is a convex optimization problem.
\end{theorem}
\begin{proof}
	It follows from Lemma \ref{lm:cost_diff}  and some algebraic manipulation that the optimization problem in \eqref{eq:delta_conv_corr} is equivalent to
\begin{equation}
	\begin{aligned}  
		&\min_{{\Deltam}} \  (1-\lambda) \log \left| \Sigmam_{Y\!Y}+\Sigmam_{i-1} + {{\Deltam}}\right| \\
		&\qquad - \log \left| \sigma^2\textnormal{\textbf{I}}_m +\Sigmam_{i-1} + {{\Deltam}} \right|  + \lambda \textnormal{tr}\left(\Sigmam^{-1}_{Y\!Y}{{\Deltam}_i}\right)\\
		&\ \textnormal{s.t.} \ \ \Deltam\in\Dc_j, \\
		&\qquad \  \Sigmam_{i-1}+\Deltam\in S_+^m. \\
	\end{aligned}
\end{equation}
Noting that the sets $\Dc_j$ are convex for all $j\in\Ac_{i-1}^{\sf c}$, that the logarithm terms are  convex \cite{boyd_cvx} for $\lambda\geq 1$, and that the trace term is linear, yields that the optimization problem in \eqref{eq:delta_conv_corr} is convex in ${{\Deltam}}$. This completes the proof.
\end{proof}
 
The proposed greedy construction for independent attack case is described in Algorithm \ref{alg:greedy_corr}. Note that the matrix obtained in the optimization problem in Theorem \ref{th:delta_conv_corr} is constrained by projecting the sum of the update and the previous covariance matrix in the positive semidefinite cone to guarantee that the resulting covariance matrix is indeed positive semidefinite. This is reflected in the last step of Algorithm \ref{alg:greedy_corr} where the resulting matrix construction is projected by minimizing the Frobenius distance to the positive semidefinite cone.

\begin{algorithm}
	\caption{$k$-sparse correlated attack construction}
	\label{alg:greedy_corr}
	\begin{algorithmic}[1]
		\Require 	$\Hm$ in \eqref{eq:obs_noattack};\newline
				$\sigma^2$ in \eqref{EqZ};\newline
				$\Sigmam_{X\!X}$ in \eqref{EqSigmaXX};\newline
				$\lambda$ in \eqref{eq:Gauss_cost}; and \newline
				$k$ in \eqref{EqSetAi}.
		\Ensure $\bm{\Sigma}_{A\!A}$ in \eqref{eq:Gauss_attack}.
		\State Set $\Ac_0=\left\{\emptyset \right\}$
		\State Set $\Sigmam_0=\mathbf{0}$
		\For {$j=1$ to $k$}
		\For {$\ell\in\Ac_{j-1}^{\sf c}$} 
		\State  Compute $\Deltam_{\ell}=\argmin_{\Deltam\in\Dc_\ell}J(\Sigmam_{j-1}+\Deltam)$ 
		\EndFor
		\State Compute $j^{\star}=\argmin_{\ell\in\Ac_{j-1}^{\sf c}}J(\Sigmam_{j-1} + {\Deltam}_\ell )$
		\State Set $\Ac_j = \Ac_{j-1}\cup \left\{j^{\star}\right\}$
		\State Set $ \Sigmam_j = \Sigmam_{j-1} + {\Deltam}_{j^\star} $
		\EndFor
		\State Compute ${\Sigmam}_{A\!A}= \argmin_{\Sm \in \Sc_+^m} \|\Sigmam_{k} - \Sm\|_{{\sf F}}$ 
	\end{algorithmic}
\end{algorithm}

\section{Numerical Results}\label{numerical results}

In this section, we numerically evaluate the performance of the proposed attack construction algorithms on a direct current (DC) state estimation setting for the IEEE 9-Bus, IEEE 14-Bus and IEEE 30-Bus test systems~\cite{UoW_ITC_99}. The voltage magnitudes are set to $1.0$ per unit, which implies that the state estimation is based on the observations of active power flow injections to all the buses and the active power flow between physically connected buses. The Jacobian matrix $\Hm$ is determined by the reactance of the branches and the topology of the corresponding systems. We use MATPOWER~\cite{matpower} to generate $\Hm$ for each test system.
The statistical dependence between the state variables is captured by a Toeplitz model for the covariance matrix $\Sigmam_{X\!X} \in \Sc_+^n$ that arises in a wide range of practical settings, such as autoregressive stationary processes \cite{IE_TSG_16}, \cite{SE_TSG_19}, \cite{ETG_TC_13}.
Specifically, we model the correlation between state variables $X_i$ and $X_j$ with the exponential decay parameter $\rho\in\mathds{R}_+$ that defines the entries of the covariance matrix of the state variables as $\left(\Sigmam_{X\!X}\right)_{ij} =  \rho^{|i-j|}$ with $(i,j) \in \{1,2,\ldots,n\}\times\{1,2,\ldots,n\}$.

In this setting, the performance of the proposed sparse stealth attack is not only a function of the attack constructions but also the correlation parameter $\rho$, the noise variance $\sigma^2$, and the topology of the system described by $\Hm$.
In the simulations, we set the observation model noise regime in terms of the signal to noise ratio (SNR) defined as
\be
\textrm{SNR} \eqdef 10\log_{10}\left(\frac{\textrm{tr}(\textbf{H}\Sigmam_{\textrm{XX}}\textbf{H}^\textrm{\sf T})}{m\sigma^{2}}\right).
\ee

\subsection{Performance in terms of information theoretic cost}
\begin{figure}[ht!]
	\centering
	\includegraphics[width=7.5cm]{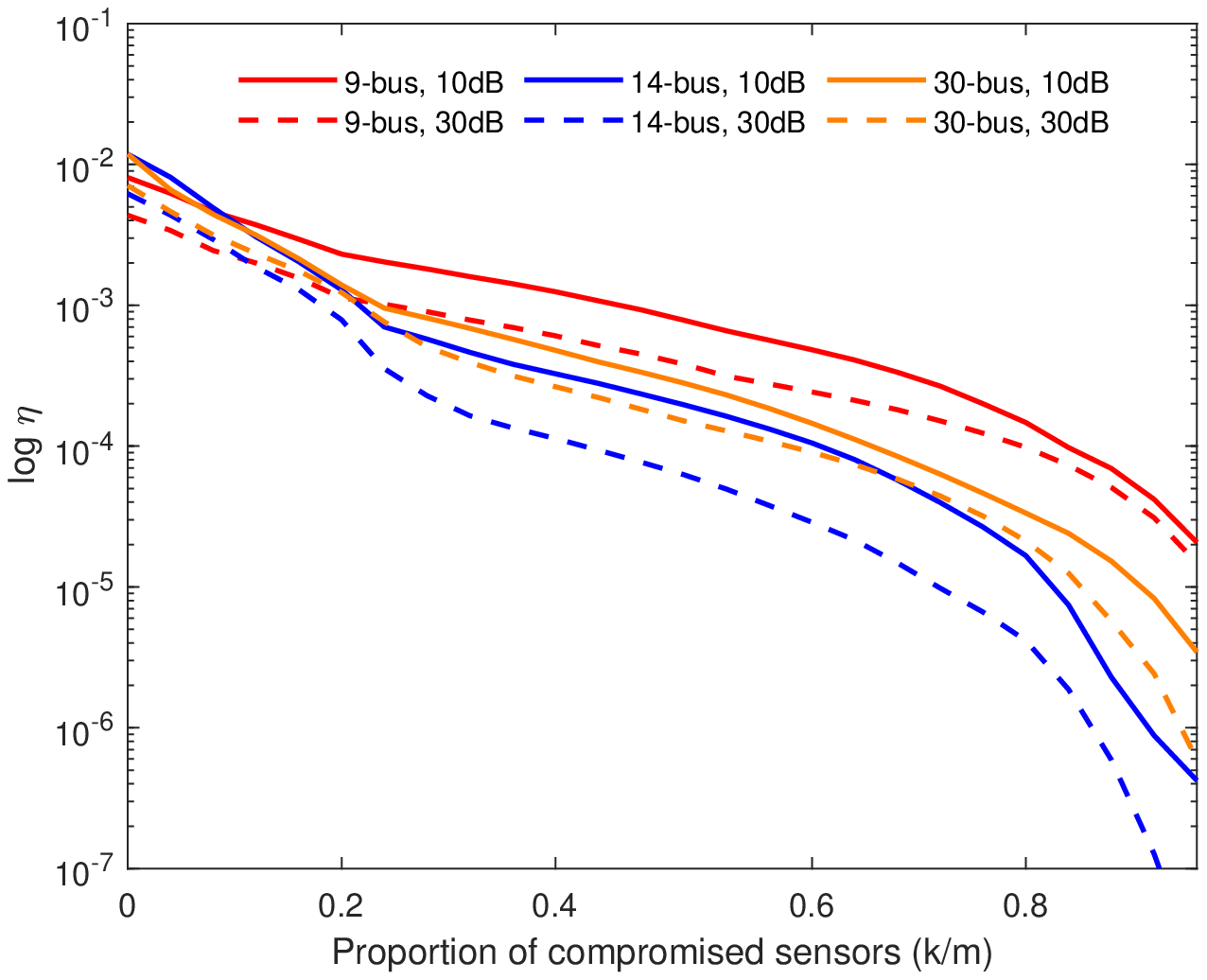}
	\caption{Performance of independent attack constructions on different IEEE test systems with $\rho = 0.9$ and $\lambda=8$.}\label{fig1_case1_rho9}
	\centering
	\includegraphics[width=7.5cm]{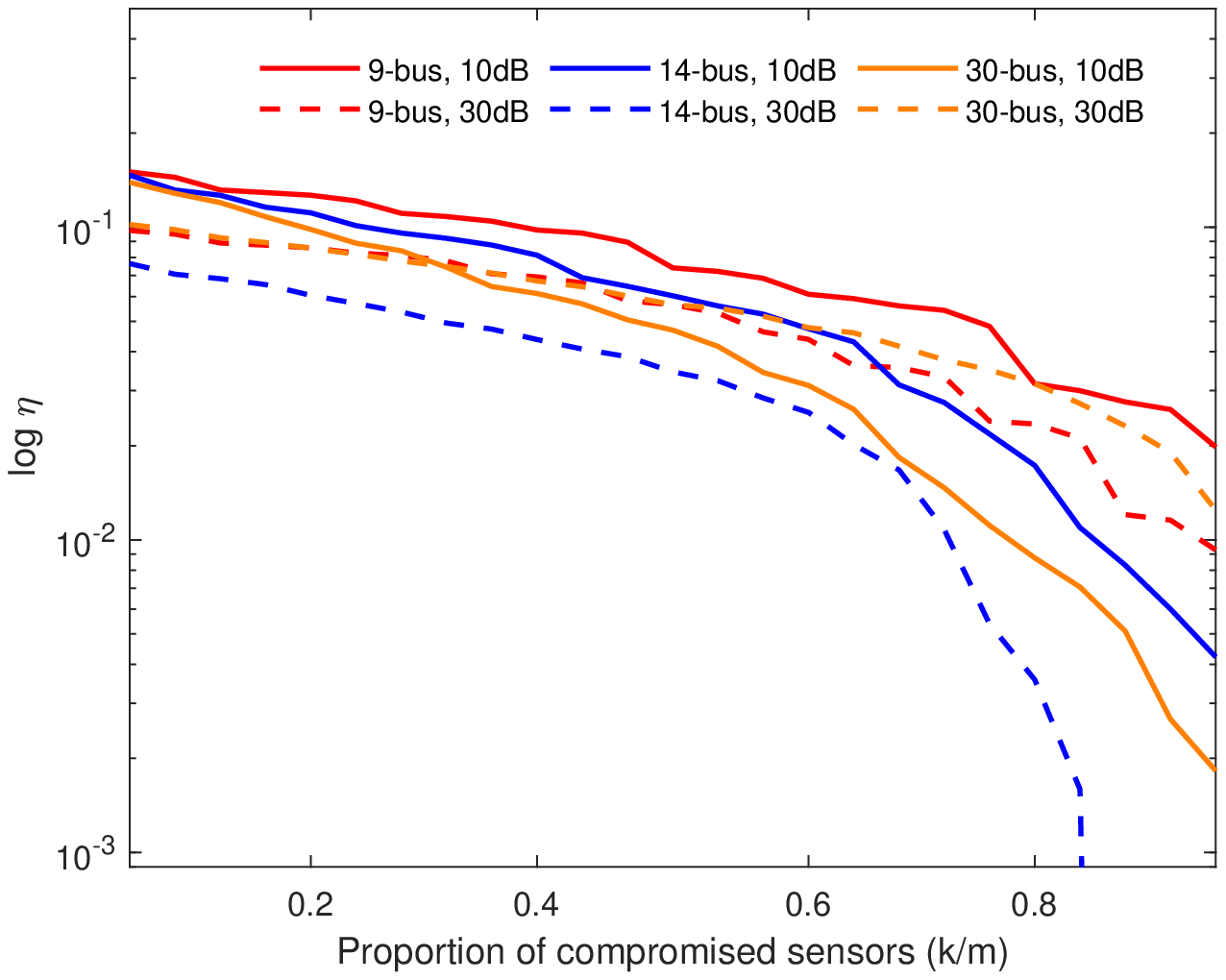}
	\caption{Performance of correlated attack constructions on different IEEE test systems with $\rho = 0.9$ and $\lambda=8$.}\label{fig1_case2_rho9}
\end{figure}
 
Let $\Sigmam_i^k$ be the output of the $k$-sparse attack construction of Algorithm $i$. 
We evaluate the attack performance in terms of the sparsity penalty defined as
\begin{equation}
	\eta \eqdef \frac{J(\Sigmam_i^k) - J(\Sigmam_i^m)}{J(\Sigmam_i^m)},
\end{equation}
where $J(\cdot)$ is the cost defined in \eqref{eq:Gauss_cost}. Note that $J(\Sigmam_i^m)$ denotes the cost induced by the construction when all the sensors are attacked. In that sense, this metric captures the performance loss of the attack when only $k$ sensors are attacked.
Fig.~\ref{fig1_case1_rho9} depicts the performance of the independent sparse stealth attack construction obtained with Algorithm \ref{alg:greedy_independent} in different IEEE test systems as a function of the proportion of compromised sensors, i.e. $k/m$, for correlation parameter $\rho = 0.9$ and $\lambda = 8$. Similarly, Fig.~\ref{fig1_case2_rho9} depicts the performance of the correlated sparse stealth attack construction from Algorithm \ref{alg:greedy_corr} in the same setting as in Fig.~\ref{fig1_case1_rho9}.
As expected, in both cases the sparsity penalty decreases monotonically with the proportion of compromised sensors. In the independent sparse attack case, the sparsity penalty does not change significantly in terms of the proportion of compromised sensors while in the Algorithm \ref{alg:greedy_corr} construction case the sparsity penalty decreases exponentially in the number of compromised sensors.
Note that the exponential decrease slope is approximately constant, which indicates that the advantage of adding more sensors to the attack construction decreases exponentially at an approximately constant rate. Remarkably, this exponential decrease is observed for all system sizes and SNR regimes.
\begin{figure}[htbp]
	\centering
	\includegraphics[width=7.5cm]{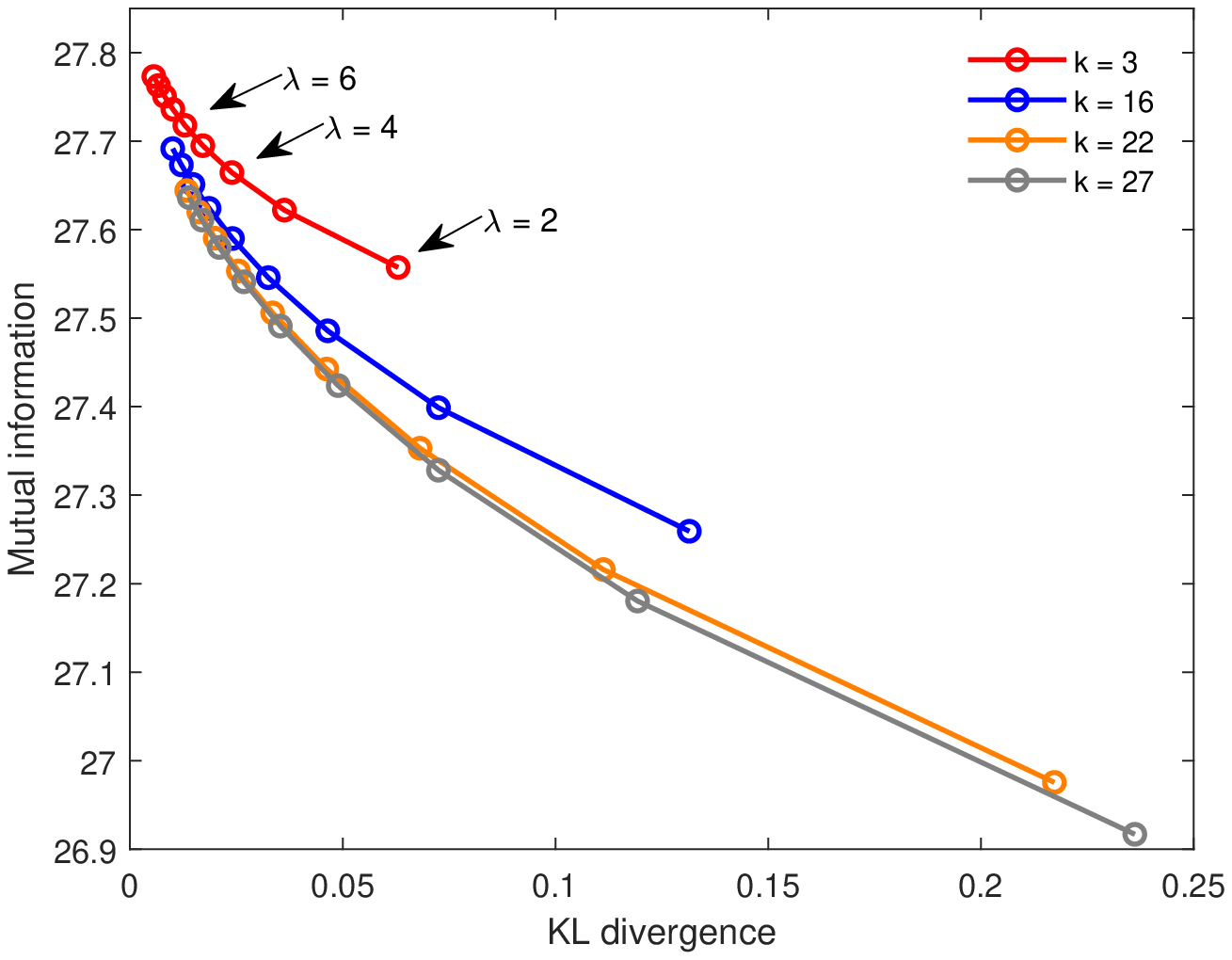}
	\caption{Performance of independent sparse attack construction in terms of mutual information and KL divergence for different values of $\lambda$ on the IEEE 9-bus system with {SNR = 30 dB and $\rho = 0.9$}.}\label{tradeoff_case1_9bus}
	\centering
	\includegraphics[width=7.5cm]{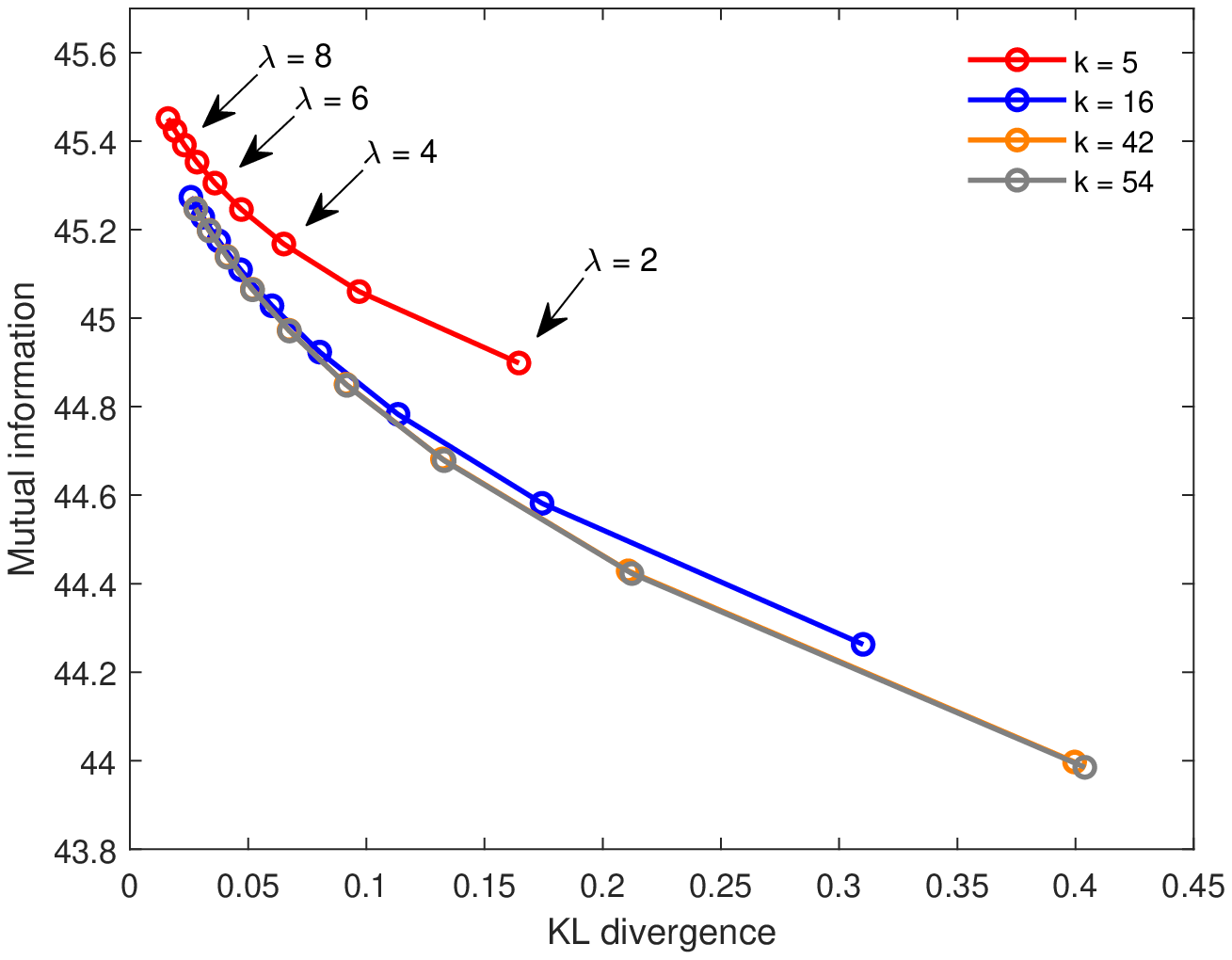}
	\caption{Performance of independent sparse attack construction in terms of mutual information and KL divergence for different values of $\lambda$ on the IEEE 14-bus system with {SNR = 30 dB and $\rho = 0.9$}.}\label{tradeoff_case1_14bus}
\end{figure}
It is worth noting that for most systems, operating with larger SNR yields a lower mutual information for the same KL divergence. However, in Fig. \ref{fig1_case2_rho9} for the IEEE 30-bus test system the $10$ dB and $30$ dB performance curves cross, which indicates that the lower SNR regime benefits the attacker when the number of comprised sensors grows. Interestingly, the size of the network does not determine the performance the attack. For the Algorithm 1  construction, the IEEE 14-bus system is the most vulnerable to attacks, while for the Algorithm 2  construction the statement only holds for high SNR regime. This suggests that the topology of the network fundamentally changes the performance of the attack but the specific mechanisms are left for future study.

\subsection{Performance in terms of the tradeoff between mutual information and KL divergence}

\begin{figure}[htbp]
	\centering
	\includegraphics[width=7.5cm]{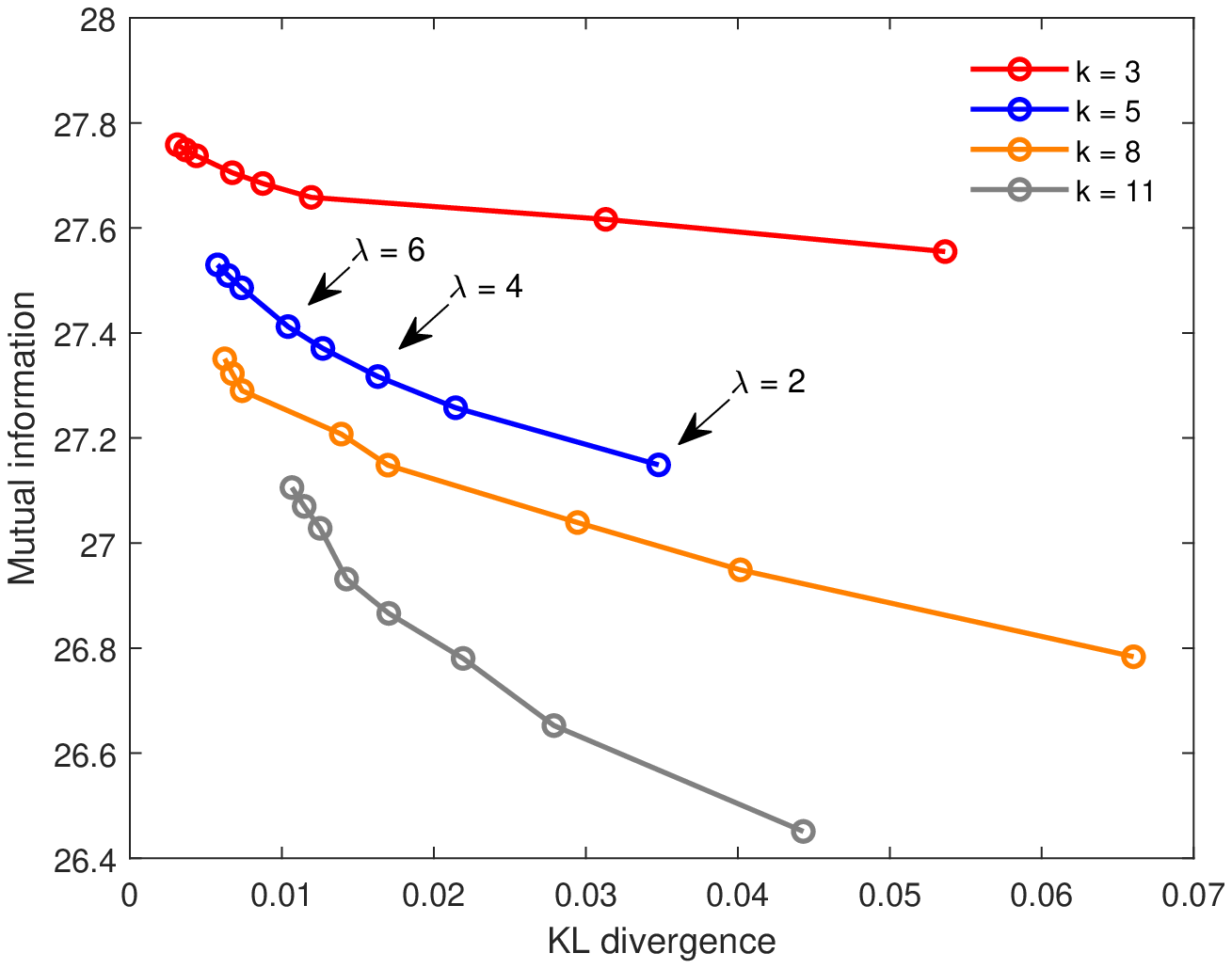}
	\caption{Performance of correlated sparse attack construction in terms of mutual information and KL divergence for different values of $\lambda$ on the IEEE 9-bus system with {SNR = 30 dB and $\rho = 0.9$}.}\label{tradeoff_case2_9bus}
	\centering
	\includegraphics[width=7.5cm]{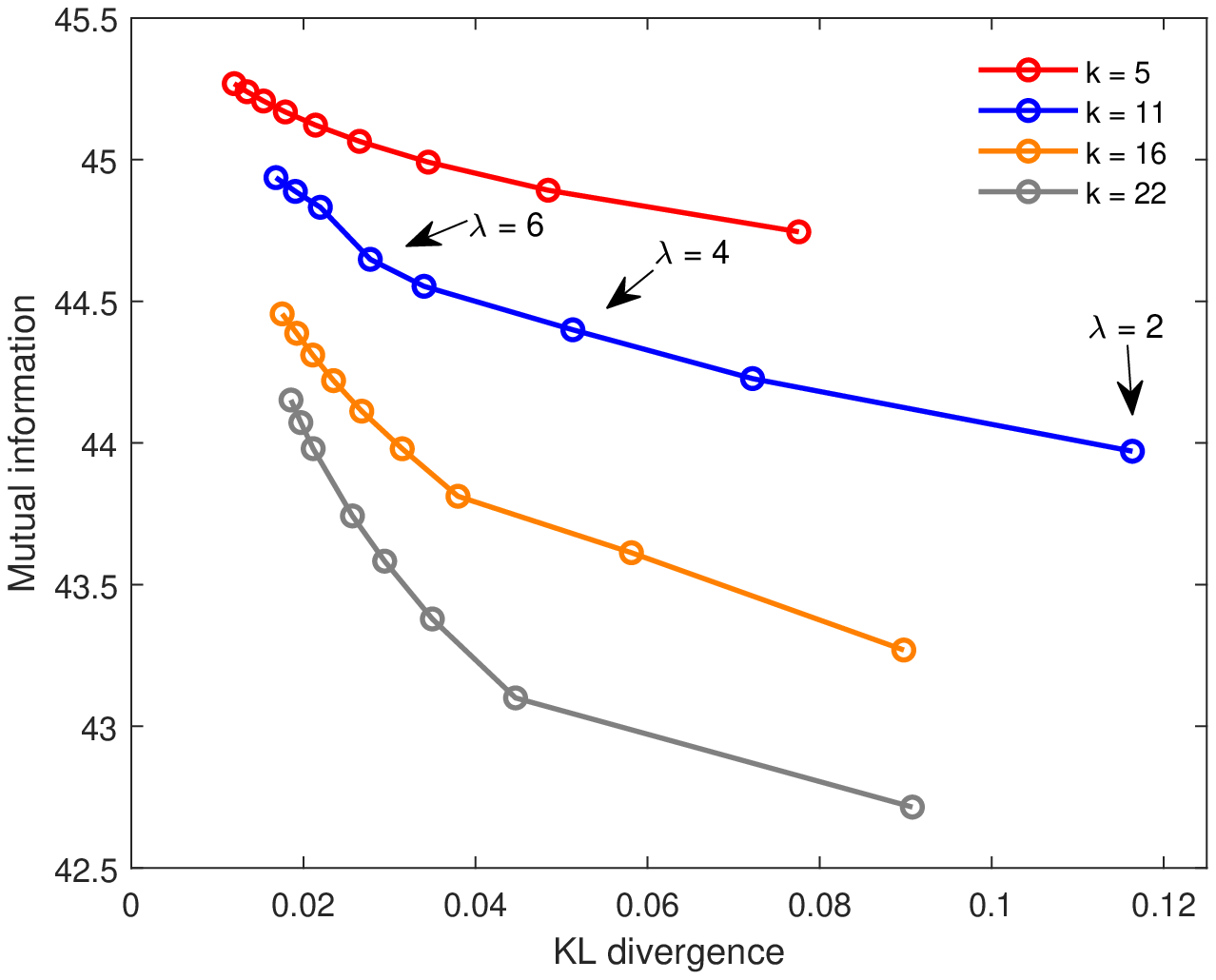}
	\caption{Performance of correlated sparse attack construction in terms of mutual information and KL divergence for different values of $\lambda$ on the IEEE 14-bus system with {SNR = 30 dB and $\rho = 0.9$}.}\label{tradeoff_case2_14bus}
\end{figure}
Fig.~\ref{tradeoff_case1_9bus}  and 
Fig.~\ref{tradeoff_case1_14bus} depict the multiobjective performance of the Algorithm \ref{alg:greedy_independent} attack construction in terms of the tradeoff between mutual information and KL divergence for different values of the proportion of compromised sensors when SNR = 30 dB and $\rho = 0.9$. Similarly, 
Fig. \ref{tradeoff_case2_9bus} and 
Fig. \ref{tradeoff_case2_14bus} depict the same setting for the  Algorithm \ref{alg:greedy_corr} attack construction.
As expected, larger values of the parameter $\lambda$ yield smaller values of KL divergence, i.e. the probability of detection is prioritized in the construction over the mutual information decrease for all the scenarios. Moreover, smaller values of $k$ yield smaller reductions of the mutual information, which indicates that remaining stealthy in a sparse setting necessarily implies reducing the amount of disruption of the attack.
On the other hand, larger values of $k$ enable the attacker to more effectively tradeoff disruption for stealth. This effect is particularly marked in the correlated attack construction case, which reinforces the previous observation regarding the value of coordination between attack variables to achieve stealth. 
\begin{figure}[htbp]
	\centering
		\includegraphics[width=7.5cm]{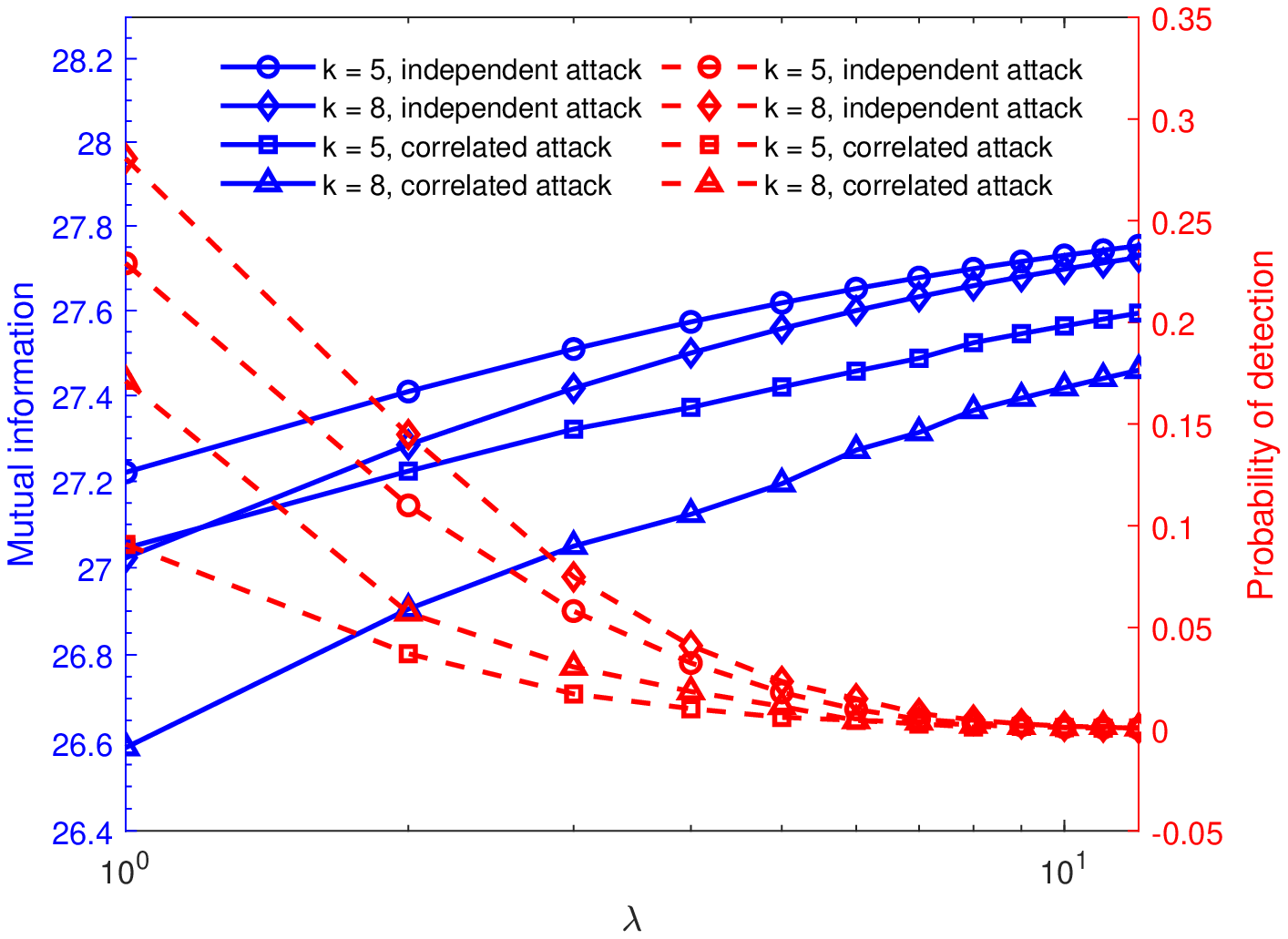}
		\caption{{Performance of attack constructions on IEEE 9-bus test system with $\rho=0.9$, SNR = 30dB and $\tau = 2$.}}\label{MI_P_bus9} 
	\centering
		\includegraphics[width=7.5cm]{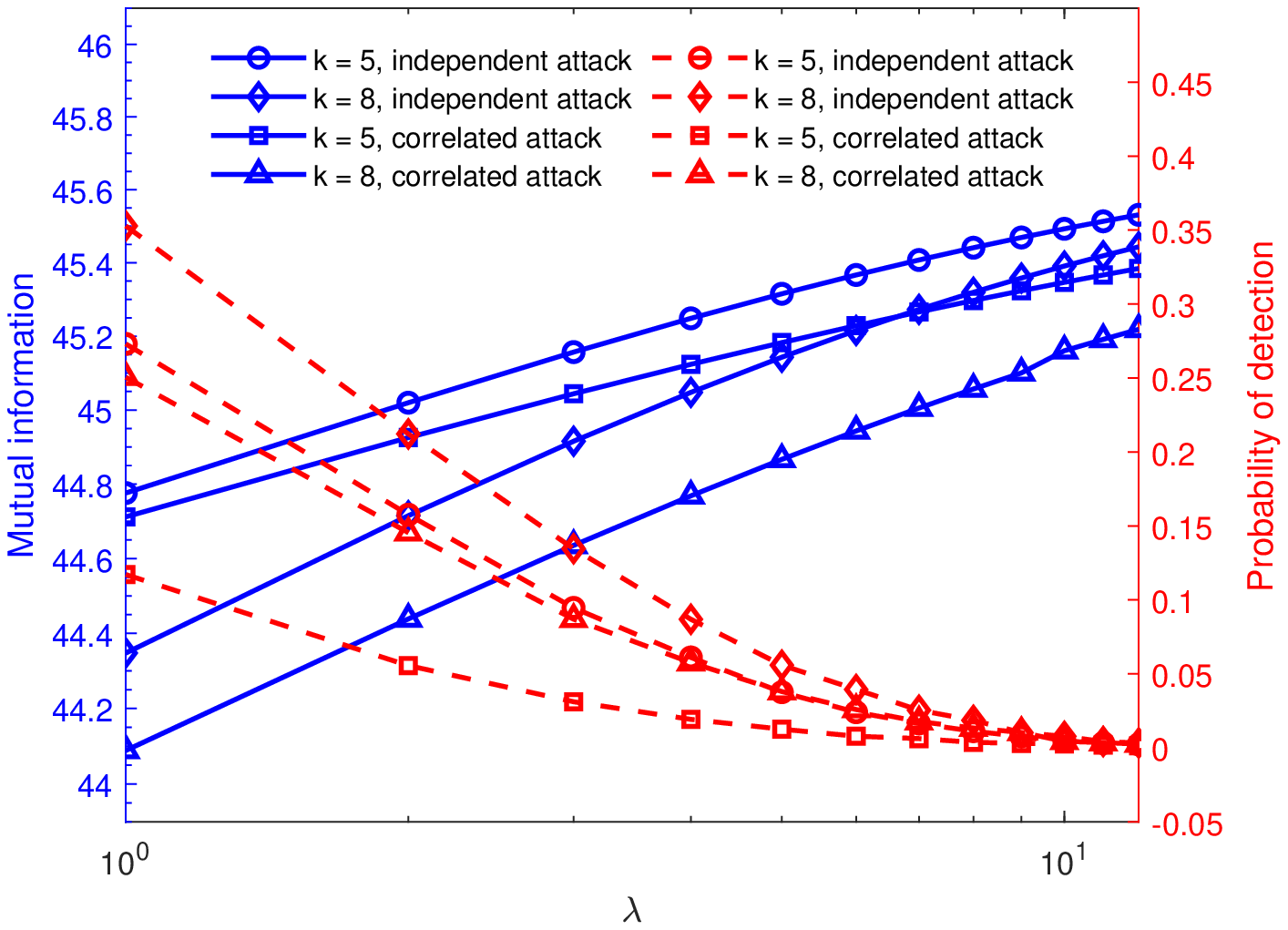}
		\caption{{Performance of attack constructions on IEEE 14-bus test system with $\rho=0.9$, SNR = 30dB and $\tau = 2$.}}\label{MI_P_bus14}
\end{figure}
{\subsection{Performance in terms of mutual information and probability of attack detection}
Fig.~\ref{MI_P_bus9} and Fig.~\ref{MI_P_bus14} depict the performance of the attack construction for different values of $\lambda$ and sparse constraint $k$ with SNR = 30 dB, $\rho = 0.9$ and $\tau = 2$ for the IEEE 9-bus and the IEEE 14-bus test systems, respectively. As expected, larger values of the parameter $\lambda$ yield smaller values of the probability of attack detection while increasing the mutual information between the vector of state variables and the vector of observations in the systems. We note that the probability of attack detection decreases approximately linearly with respect to $\log \lambda$ for small values of $\lambda$. Simultaneously for this range of $\lambda$, mutual information increases approximately linearly with respect to $\log \lambda$. For moderate values of $\lambda$, we observe a significant decrease in the probability of detection with respect to $\log \lambda$ with a smaller rate of increase in mutual information. The comparison between independent and correlated attack constructions, shows that for the same sparsity constraint, the correlated attack construction successfully exploits the coordination between different locations to yield a smaller probability of detection and a smaller mutual information.}

\section{Conclusion}\label{conclusion}

We have proposed novel stealth attack construction with sparsity constraints. The insight obtained from the problem of incorporating an additional sensor to the attack has been distilled to construct heuristic greedy constructions for both the independent and the correlated attack cases. We show that for both cases, the greedy step results in a convex optimization problem which can be solved efficiently and yields a low complexity attack update rule. We have numerically evaluated the attack performance in several IEEE test systems and shown that it is feasible to implement disruptive attacks that have access to small number of observations. Furthermore, we have observed that the topology and the SNR regime govern the performance of the attack and numerically characterized the dependence.


\balance
 
\bibliographystyle{IEEEtran}
\bibliography{TSG_21}

\end{document}